\documentclass[a4paper]{article}
\usepackage{amsmath,amssymb,amsthm,algorithm,enumitem,pgfplots,pgfplotstable,tikz}
\usepackage[noend]{algorithmic}
\usetikzlibrary{calc}

\newtheorem{lemma}{Lemma}
\newtheorem{corollary}{Corollary}
\newtheorem{proposition}{Proposition}
\newtheorem{theorem}{Theorem}
\theoremstyle{definition}
\newtheorem{definition}{Definition}
\theoremstyle{remark}
\newtheorem{example}{Example}
\newtheorem{problem}{Open Problem}

\newcommand{\GETS}{:=}
\newcommand{\PROCEDURE}[1]{\item[]\quad\item[\textbf{Procedure}] {#1}\textbf{:}}
\newcommand{\VAR}[1]{\mathop{\textit{#1}}}

\DeclareMathOperator{\lcm}{lcm}
\newcommand\union{\cup}
\newcommand\intersect{\cap}
\newcommand\transitive[1]{{#1}^*}
\newcommand\concat{\cdot}
\newcommand\abs[1]{{\lvert #1 \rvert}}
\newcommand\slen[1]{{\lVert #1\rVert}}
\newcommand\mlanglen[1]{\mathop{\textsf{min}}{\lVert #1 \rVert}}

\newcommand\Alphabet{\Sigma}
\newcommand\Lang{\mathcal{L}}
\newcommand\String{S}
\renewcommand\symbol{\sigma}
\newcommand\EmptyString{\varepsilon}

\newcommand\Graph{\mathsf{G}}
\newcommand\Nodes{\mathsf{Q}}
\newcommand\Transitions{\delta}
\newcommand\TransTransitions{\delta^*}
\newcommand\Path{\pi}
\newcommand\Trace{\mathop{\textsf{trace}}}

\newcommand\Grammar{\mathbf{C}}
\newcommand\NonTerminals{\mathbf{N}}
\newcommand\Productions{\mathbf{P}}
\newcommand\NT[1]{\texttt{#1}}
\newcommand\ruleto{\mapsto}
\newcommand\produces[1]{\rightarrow^*_{#1}}
\newcommand\pproduces[1]{\rightarrow^+_{#1}}
\newcommand\ANT[3]{#1[#2,#3]}
\newcommand{\heads}{\mathop{\textsf{heads}}}
\newcommand\NTS[2]{\langle #1 \rangle_{#2}}
\newcommand{\NDString}[2]{\mathop{\textsf{string}}(#1; #2)}

\newcommand{\Anno}[2]{#1_{#2}}

\newcommand{\BigO}{\mathcal{O}}

\pgfplotscreateplotcyclelist{black white var}{%
    every mark/.append style={thick,fill=gray},mark=triangle*\\%
    every mark/.append style={thick,fill=gray},mark=diamond*\\%
}
\pgfplotsset{
    plot/.style={width=0.375\textwidth,
                 cycle list name=black white var,
                 every axis legend/.append style={
                    draw=none,
                    cells={anchor=west},
                    at={(0.5,-0.50)},
                    anchor=south
                 }
                }
}

\title{Querying for Paths in Graphs using Context-Free Path Queries}
\author{Jelle Hellings\\Hasselt University\\
        \texttt{jelle.hellings@uhasselt.be}
      }

\begin{document}

\maketitle

\begin{abstract}
Navigational queries for graph-structured data, such as the regular path queries and the context-free path queries, are usually evaluated to a relation of node-pairs $(m, n)$ such that there is a path from $m$ to $n$ satisfying the conditions of the query. Although this \emph{relational query semantics} has practical value, we believe that the relational query semantics can only provide limited insight in the structure of the graph data. To address the limits of the relational query semantics, we introduce the \emph{all-path query semantics} and the \emph{single-path query semantics}. Under these path-based query semantics, a query is evaluated to all paths satisfying the conditions of the query, or, respectively, to a single such path.

While focusing on context-free path queries, we provide a formal framework for evaluating queries on graphs using both path-based query semantics. For the all-path query semantics, we show that the result of a query can be represented by a finite context-free grammar annotated with node-information relevant for deriving each path in the query result. For the single-path query semantics, we propose to search for a path of minimum length. We reduce the problem of finding such a path of minimum length to finding a string of minimum length in a context-free language, and for deriving such a string we propose a novel algorithm.

Our initial results show that the path-based query semantics have added practical value and that query evaluation for both path-based query semantics is feasible, even when query results grow very large. For the single-path query semantics, determining strict worst-case upper bounds on the size of the query result remains the focus of future work.
\end{abstract}

\section{Introduction}
The graph data model is one of the most versatile and natural data models in use: graph-structured data is everywhere and examples can be found in family trees, social networks, process models, gene networks, XML data, and RDF data~\cite{genenetwork,xml,modelcheck,rdf}. For querying graphs, many different query languages have been developed, proposed, and researched~\cite{rpq,ecrpq,xpath2,xpath1,crpq,relexpr,sparql,icdt2014}. At their core, most graph query languages depend on navigating the graph. This graph navigation is usually performed by means of a regular expression that describes the allowed edge-labeling of the paths that should be traversed in the graph. As the regular expressions have limited expressive power, we focus on a more expressive navigational query language, namely the context-free path queries that use context-free grammars to describe the labeling of paths~\cite{harel,icdt2014,pdl,biocfg}.

These navigational queries expressed by context-free path queries are usually evaluated to a relation of node-pairs $(m, n)$ such that there is a path from $m$ to $n$ whose labeling is described by a context-free grammar---the \emph{relational query semantics}, or to the truth value true whenever such a path exists---the \emph{boolean query semantics}. Although many practical problems can be answered by navigational queries evaluated under the usual semantics, we believe that the relational query semantics and the boolean query semantics are limiting. The inability to view the paths of interest hampers the understanding of the data, makes query debugging harder, and makes it impossible to answer certain practical problems.

To address the limitations of the traditional query semantics, we introduce path-based query semantics. Concretely, we introduce the \emph{all-path query semantics} and the \emph{single-path query semantics}. Under the all-path query semantics, a query is evaluated to all paths satisfying the conditions of the query, and under the single-path query semantics one such path is chosen. The  practical usage of these path-based query semantics can be illustrated by a simple example:

\begin{example}\label{exam:same_generations}
Consider a collection of family trees represented by a graph in which the nodes represent peoples and the edges represent \textit{parentOf} and \textit{childOf} relations (between parents and their children). Consider the context-free grammar with the following production rules: 
\begin{align*}\NT{q} &\ruleto \textit{parentOf}\ \NT{q}\ \textit{childOf},&
              \NT{q} &\ruleto \textit{parentOf}\ \textit{childOf}.\end{align*}

Using the standard relational query semantics, the query $\NT{q}$ evaluates to the relation of node-pairs $(m, n)$ such that $m$ and $n$ are both $k$-th generation descendants of a common ancestor. Using the single-path query semantics that we propose, the query $\NT{q}$ evaluates to a path from $m$ to a common ancestor and from this common ancestor to $n$, showing why $m$ and $n$ are both $k$-th generation descendants of a common ancestor, while, at the same time, showing who this common ancestor is.
\end{example}

Observe that the context-free grammar used in Example~\ref{exam:same_generations} is well-known to not be expressible by a regular expression~\cite{flbook}. Still, this simple example is at the basis of practical queries that are used in, for example, bio-informatics~\cite{biocfg}.

For graph querying, path-based query semantics have only gained limited attention. For the regular expressions, Barceló et al.~\cite{ecrpq} introduced the extended regular path queries that have path variables for output. The main focus of Barceló et al. is, however, on the use of path variables for expressivity purposes, and path-based results are only studied in limited details. Recent work by Hofman et al.~\cite{seper_strings} provides an alternative to use path-based query semantics for debugging:  to gain more insight in the behavior of regular path queries with respect to the expected behavior, Hofman et al. propose a technique based on separability. Although this approach addresses query debugging, it does not lift the other limitations of the relational and the boolean query semantics.

In the setting of model checking using CTL~\cite{modelcheck}, path-based query semantics are widely used. Normally, CTL formulae are evaluated to true or false, indicating if the graph meets or not meets certain conditions. An important ability of CTL model checking algorithms is to not only answer CTL formulae with a truth value, but to also answer with a witnesses or a counterexample for this truth value. These witnesses and counterexamples are represented by a path in the graph that shows why the graph does or doesn't meet the conditions expressed by the CTL formulae. Counterexamples and witnesses also exists for other modal logics, such as LTL. These path-based witnesses and counterexamples are especially useful in the analysis of the model checking results.

In this work we study path-based query semantics. We provide a formal framework for evaluating queries on graphs using the all-path query semantics and the single-path query semantics. To achieve this, we first show how to represent the query result under the all-path query semantics by a context-free grammar annotated with node-information, and we show that this context-free grammar can be used to derive exactly those paths that are in the query result.

For the single-path query semantics, we propose to search for a path of minimum length. As we can represent the set of all paths by an annotated context-free grammar, we reduce the problem of finding a path of minimum length matching the query conditions to finding a string of minimum length in a context-free language. For deriving such a string of minimum length, we propose a novel algorithm. We then proceed with the analysis of this minimum-length string derivation algorithm applied to annotated context-free grammars by analyzing the possible length of minimum-length paths. For annotated context-free grammars over the singleton alphabet, we show a close-to-strict worst-case upper bound on the length of minimum-length paths that is linear in the number of nodes in the graph. For general annotated context-free grammars we show that the worst-case upper bound on the length of minimum-length paths is at least quadratic in the number of nodes in the graph.

To test the behavior of the minimum-length path derivation algorithm in practice, we performed measurements on an initial implementation. These results show promise, as the initial implementation shows acceptable performance for a range of context-free path queries.

\paragraph*{Organization} 
In Section~\ref{sec:prelim}, we present the basic notions used throughout this paper. In Section~\ref{sec:cfpq}, we present the context-free path queries together with their usual semantics, and we introduce the all-path and single-path query semantics. In Section~\ref{sec:all_path} and Section~\ref{sec:single_path} we introduce approaches to evaluate queries using the all-path query semantics and the single-path query semantics, respectively. In Section~\ref{sec:exp}, we present our results on a small-scale implementation. In Section~\ref{sec:conclusions}, we summarize our findings and propose directions for future work.

\section{Preliminaries}\label{sec:prelim}
We call a sequence $\symbol_1 \dots \symbol_n$ of symbols a string. The length of string $\String = \symbol_1 \dots \symbol_n$, denoted by $\slen{\String}$, is $n$. The empty string is denoted by $\EmptyString$ and we usually treat individual symbols as strings of length one. The concatenation of two strings $\String_1$ and $\String_2$ is denoted by $\String_1 \concat \String_2$. If $\Alphabet$ is a set of symbols, then we denote the set of all strings made of symbols from $\Alphabet$ by $\transitive{\Alphabet}$.

\begin{definition}
A \emph{graph} is a triple $\Graph = (\Nodes, \Alphabet, \Transitions)$ with $\Nodes \intersect \Alphabet = \emptyset$, in which $\Nodes$ is a finite set of nodes, $\Alphabet$ is a finite set of alphabet symbols used as edge labels, and $\Transitions \subseteq \Nodes \times \Alphabet \times \Nodes$ is a finite set of labeled edges. 

If $m \in \Nodes$ and $\symbol \in \Alphabet$, then $\Transitions(m, \symbol) = \{ n \mid (m, \symbol, n) \in \Transitions \}$ denotes those nodes that have an incoming edge labeled with $\symbol$ originating at $m$. If $m \in \Nodes$ and $\String \in \transitive{\Alphabet}$, then \[\TransTransitions(m, \String) = \begin{cases} \{ m \} & \text{if $\String = \EmptyString$};\\
		\bigcup_{n \in \Transitions(m, \symbol)} \TransTransitions(n, \String') &\text{if $\String = \symbol \concat \String'$}.
\end{cases}\] The \emph{language of graph} $\Graph$ with respect to $m, n \in \Nodes$, denoted by $\Lang(\Graph; m, n)$, is defined by \[\Lang(\Graph; m, n) = \{ \String \mid \String \in \transitive{\Alphabet} \land n \in \TransTransitions(m, \String) \}.\]
\end{definition}

Let $\Graph = (\Nodes, \Alphabet, \Transitions)$ be a graph. A \emph{path} $\Path = n_1\symbol_1\dots n_{i-1} \symbol_{i-1} n_i$ in $\Graph$ is a sequence with, for all $1 \leq j < i$, $(n_j, \symbol_j, n_{j+1}) \in \Transitions$. We write $n_1 \Path n_i$ to indicate that $\Path$ starts at node $n_1$ and ends at node $n_i$. The \emph{trace} of $\Path$ is defined by $\Trace(\Path) = \symbol_1\dots \symbol_{i-1}$. Observe that traces are strings over the alphabet $\Alphabet$.

\begin{definition}
A \emph{context-free grammar} is a triple $\Grammar = (\NonTerminals, \Alphabet, \Productions)$  with $\NonTerminals \intersect \Alphabet = \emptyset$, in which $\NonTerminals$ is a set of non-terminals, $\Alphabet$ is a finite set of alphabet symbols, and $\Productions$ is a set of production rules.\footnote{Usually, context-free grammars are defined with a dedicated start non-terminal. It is straightforward to specialize our results on context-free grammars to the setting with a dedicated start non-terminal.} In the above, a production rule is of the form $\NT{a} \ruleto \NT{b}\ \NT{c}$ or $\NT{a} \ruleto \symbol$, in which $\NT{a}, \NT{b}, \NT{c} \in \NonTerminals$ and $\symbol \in \Alphabet$.\footnote{To simplify the presentation, we assume that context-free grammars are in Chomsky Normal Form~\cite{flbook}, and we exclude the derivation of $\EmptyString$. Unless stated otherwise, it is straightforward to generalize our results on context-free grammars to the setting that includes production rules of the form $\NT{a} \ruleto \EmptyString$.}

Production rules are to be interpreted as rewrite rules:  if $\String = \String_1 \concat \NT{a} \concat \String_2$ is a string with $\String_1, \String_2 \in \transitive{(\NonTerminals \union \Alphabet)}$ and $\NT{a} \in \NonTerminals$, and if $(\NT{a} \ruleto \String') \in \Productions$, then $\String$ can be rewritten into $\String_1 \concat \String' \concat \String_2$ by application of $\NT{a} \ruleto \String'$. We write $\String \produces{\Productions} \String'$ if $\String$ can be rewritten into $\String'$ by a finite number of rewrites using production rules in $\Productions$ and we write $\String \pproduces{\Productions} \String'$ if $\String \produces{\Productions} \String'$ and at least one rewrite step is necessary to rewrite $\String$ into $\String'$. The \emph{language of a context-free grammar} $\Grammar$ with respect to $\NT{a} \in \NonTerminals$, denoted by $\Lang(\Grammar; \NT{a})$, is defined by \[\Lang(\Grammar; \NT{a}) = \{ \String \mid \String \in \transitive{\Alphabet} \land \NT{a} \produces{\Productions} \String \}.\]
\end{definition}

\section{Context-free path queries}\label{sec:cfpq}

Let $\Grammar = (\NonTerminals, \Alphabet, \Productions)$ be a context-free grammar with $\NT{a} \in \NonTerminals$. We say that $\NT{a}$ is a \emph{context-free path query}. Usually, these queries are evaluated using the \emph{boolean query semantics} or the \emph{relational query semantics}.\footnote{Commonly, relational query semantics is referred to as path query semantics~\cite{relexpr}. To avoid confusion with our path-based query semantics, we have chosen for a different naming in this paper.} 
\begin{enumerate}
\item Using \emph{boolean query semantics}, the query $\NT{a}$ on graph $\Graph$ evaluates to the truth value of $\exists m\exists n\  \Lang(\Grammar; \NT{a}) \intersect \Lang(\Graph; m, n) \neq \emptyset$.
\item Using \emph{relational query semantics}, the query $\NT{a}$ on graph $\Graph$ evaluates to the binary relation $\{ (m, n) \mid \Lang(\Grammar; \NT{a}) \intersect \Lang(\Graph; m, n) \neq \emptyset \}$.
\end{enumerate}
We study two alternative ways of evaluating queries on graphs: the \emph{all-path query semantics} and the \emph{single-path query semantics}:
\begin{enumerate}[resume]
\item Using \emph{all-path query semantics}, the query $\NT{a}$ on graph $\Graph$, with respect to nodes $m, n \in \Nodes$, evaluates to the set of all paths $m\Path n$ in $\Graph$ with $\Trace(\Path) \in \Lang(\Grammar; \NT{a})$.
\item Using \emph{single-path query semantics}, the query $\NT{a}$ on graph $\Graph$, with respect to nodes $m, n \in \Nodes$, evaluates to a single path $m\Path n$ in $\Graph$ with $\Trace(\Path) \in \Lang(\Grammar; \NT{a})$ (if such a path exists).
\end{enumerate}

The following example illustrates the usages of these query semantics.

\begin{example}
Let $\Graph$ be a collection of family trees in which the nodes represent people and the edges represent $\textit{familyOf}$ relations (between parents and their children). We have the context-free grammar with the following production rules:
\begin{align*}\NT{q} &\ruleto \textit{familyOf},&
              \NT{q} &\ruleto \NT{q}\ \NT{q}.\end{align*}
Depending on the semantics used, the query $\NT{q}$ evaluated on $\Graph$ answers various questions:
\begin{enumerate}
\item Using boolean query semantics: \emph{`are there family members in these family trees?'}
\item Using relational query semantics: \emph{`provide all pairs of people that are related.'}
\item Using all-path query semantics: \emph{`provide every way in which $m$ and $n$ are related.'}
\item Using single-path query semantics: \emph{`provide a proof that $m$ and $n$ are related,'} or \emph{`show how $m$ and $n$ are related.'}
\end{enumerate}
\end{example}

\section{Answering queries using all-path query semantics}\label{sec:all_path}
If a context-free path query is evaluated on cyclic graphs, then the query result can be an infinite set of paths. Hence, before we look into how to answer a query using the all-path query semantics, we need to determine how to represent such an infinite set of paths using a finite structure. Graphs are strongly related to finite automata and it is well-known that the intersection of the language of a finite automaton and the language of a context-free grammar is itself a language that can be represented by a context-free grammar:

\begin{lemma}[Bar-Hillel et al.~\cite{intpars}]\label{lem:cfgintgraph}
Let $\Grammar = (\NonTerminals, \Alphabet, \Productions)$ be a context-free grammar with $\NT{a} \in \NonTerminals$ and let $\Graph = (\Nodes, \Alphabet, \Transitions)$ be a graph with $m, n \in \Nodes$. The language $\Lang(\Grammar; \NT{a}) \intersect \Lang(\Graph; m, n)$ can be represented by a context-free grammar.
\end{lemma}

Lemma~\ref{lem:cfgintgraph} only guarantees that there is a finite representation of the set of all traces of paths $m\Path n$ in graph $\Graph$ with $\Trace(\Path) \in \Lang(\Grammar; \NT{a})$. As several paths can have the same trace, the set of traces cannot be directly mapped to a set of paths. To allow for a direct representation of the set of paths, we show how to construct a context-free grammar that is annotated with node-information relevant for the derivation of paths.

\begin{definition}\label{def:anno_grammar}
Let $\Grammar = (\NonTerminals, \Alphabet, \Productions)$ be a context-free grammar and let $\Graph = (\Nodes, \Alphabet, \Transitions)$ be a graph. We denote triples $(\NT{a}, m, n) \in \NonTerminals \times \Nodes^2$ by $\ANT{\NT{a}}{m}{n}$. An \emph{annotated grammar} over $(\Grammar, \Graph)$ is a context-free grammar $\Anno\Grammar\Graph = (\Anno\NonTerminals\Graph, \Alphabet, \Anno\Productions\Graph)$ in which $\Anno\NonTerminals\Graph \subseteq \NonTerminals \times \Nodes^2$; each production rule in $\Anno\Productions\Graph$ is of the form $\ANT{\NT{a}}{m}{n} \ruleto \ANT{\NT{b}}{m}{o}\ \ANT{\NT{c}}{o}{n}$ or $\ANT{\NT{a}}{m}{n} \ruleto \symbol$, with $\NT{a}, \NT{b}, \NT{c} \in \NonTerminals$, $m, n, o \in \Nodes$, and $\symbol \in \Alphabet$; and that satisfies the following three properties:
\begin{enumerate}
    \item \label{def:anno_grammar:nonterminals} $\ANT{\NT{a}}{m}{n} \in \Anno\NonTerminals\Graph$ if and only if $\Lang(\Grammar; \NT{a}) \intersect \Lang(\Graph; m, n) \neq \emptyset$,
    \item \label{def:anno_grammar:bproductions} $(\ANT{\NT{a}}{m}{n} \ruleto \ANT{\NT{b}}{m}{o}\ \ANT{\NT{c}}{o}{n}) \in \Anno\Productions\Graph$ if and only if $(\NT{a} \ruleto \NT{b}\ \NT{c}) \in \Productions$,
    \item \label{def:anno_grammar:sproductions} $(\ANT{\NT{a}}{m}{n} \ruleto \symbol) \in \Anno\Productions\Graph$ if and only if $(m, \symbol, n) \in \Transitions$ and $(\NT{a} \ruleto \symbol) \in \Productions$.
\end{enumerate}
We say that a non-terminal $\ANT{\NT{a}}{m}{n} \in \Anno\NonTerminals\Graph$ can \emph{derive} path $\Path = n_1 \symbol_1\dots \symbol_{i-1} n_{i}$ if it can derive the string $\String = \symbol_1\dots \symbol_{i-1}$ such that, for each $1\leq j < i$, the rewrite step producing $\symbol_j$ used a production rule of the form $(\ANT{\NT{b}}{n_j}{n_{j+1}} \ruleto \symbol_j) \in \Anno\Productions\Graph$.
\end{definition}

We illustrate the concept of a annotated grammar with an example:

\begin{example}\label{exam:social_network}
Let $\Graph$ be the social network visualized in Figure~\ref{fig:exam:friends} in which the nodes represent people and the edges represent $\textit{friendOf}$ relations. Alice wants to know how she can contact Eve via friends, via friends of friends, and so on. Hence, she writes a context-free grammar $\Grammar$ with the following production rules $\Productions$: 
 \begin{align*}\NT{q} &\ruleto \textit{friendOf},& \NT{q} &\ruleto \NT{q}\ \NT{q}.\end{align*}

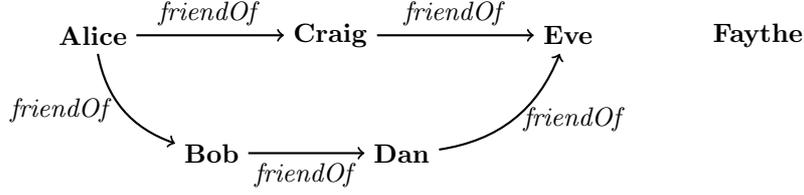
\begin{figure}[t!]
    \centering
        \begin{tikzpicture}[scale=1.25]
            \node (n1) at (  0,  0) {\textbf{Alice}};
            \node (n2) at (1.25, -1.25) {\textbf{Bob}} edge[thick,<-, bend left] node[left] {\textit{friendOf}} (n1);
            \node (n3) at (2.5, 0) {\textbf{Craig}} edge[thick,<-] node[above] {\textit{friendOf}} (n1);;
            \node (n4) at (3.25, -1.25) {\textbf{Dan}}  edge[thick,<-] node[below] {\textit{friendOf}} (n2);
            \node (n5) at (5, 0) {\textbf{Eve}} edge[thick,<-] node[above] {\textit{friendOf}} (n3)
                                         edge[thick,<-, bend left] node[right] {\textit{friendOf}} (n4);
            \node (n6) at (7, 0) {\textbf{Faythe}};       
        \end{tikzpicture}
        \caption{A social network in which persons (represented by nodes) can have \emph{friendOf}-relations (represented by labeled edges).}\label{fig:exam:friends}
\end{figure}

For brevity, we refer to each person by the first letter of their name. The annotated grammar over $(\Grammar, \Graph)$ has the following non-terminals:
\[\ANT{\NT{q}}{\text{A}}{\text{B}}, \ANT{\NT{q}}{\text{A}}{\text{C}}, \ANT{\NT{q}}{\text{A}}{\text{D}}, \ANT{\NT{q}}{\text{A}}{\text{E}},\ANT{\NT{q}}{\text{B}}{\text{D}}, \ANT{\NT{q}}{\text{B}}{\text{E}}, \ANT{\NT{q}}{\text{C}}{\text{E}}, \ANT{\NT{q}}{\text{D}}{\text{E}}.\]
The production rules $\Anno\Productions\Graph$ of the annotated grammar consists of the production rules that correspond to \textit{friendOf}-edges in the social network:
\begin{align*}
    \ANT{\NT{q}}{\text{A}}{\text{B}} &\ruleto \textit{friendOf},&
    \ANT{\NT{q}}{\text{A}}{\text{C}} &\ruleto \textit{friendOf},\\
    \ANT{\NT{q}}{\text{B}}{\text{D}} &\ruleto \textit{friendOf},&
    \ANT{\NT{q}}{\text{C}}{\text{E}} &\ruleto \textit{friendOf},\\
    \ANT{\NT{q}}{\text{D}}{\text{E}} &\ruleto \textit{friendOf}.
    \intertext{Furthermore, the following production rules of the annotated grammar express the combination of paths in the social network to form bigger paths:}
    \ANT{\NT{q}}{\text{A}}{\text{D}} &\ruleto \ANT{\NT{q}}{\text{A}}{\text{B}}\ \ANT{\NT{q}}{\text{B}}{\text{D}},&
    \ANT{\NT{q}}{\text{A}}{\text{E}} &\ruleto \ANT{\NT{q}}{\text{A}}{\text{B}}\ \ANT{\NT{q}}{\text{B}}{\text{E}},\\
    \ANT{\NT{q}}{\text{A}}{\text{E}} &\ruleto \ANT{\NT{q}}{\text{A}}{\text{C}}\ \ANT{\NT{q}}{\text{C}}{\text{E}},&
    \ANT{\NT{q}}{\text{A}}{\text{E}} &\ruleto \ANT{\NT{q}}{\text{A}}{\text{D}}\ \ANT{\NT{q}}{\text{D}}{\text{E}},\\
    \ANT{\NT{q}}{\text{B}}{\text{E}} &\ruleto \ANT{\NT{q}}{\text{B}}{\text{D}}\ \ANT{\NT{q}}{\text{D}}{\text{E}}.
\end{align*}
To produce a path from Alice to Eve, we can use this annotated grammar:
\begin{align*}
    &\ANT{\NT{q}}{\text{Alice}}{\text{Eve}}\\
    \produces{\Anno\Productions\Graph}&\text{\{Rewrite $\ANT{\NT{q}}{\text{Alice}}{\text{Eve}} \ruleto \ANT{\NT{q}}{\text{Alice}}{\text{Bob}}\ \ANT{\NT{q}}{\text{Bob}}{\text{Eve}}$\}}\\
        &\ANT{\NT{q}}{\text{Alice}}{\text{Bob}}\ \ANT{\NT{q}}{\text{Bob}}{\text{Eve}}\\
    \produces{\Anno\Productions\Graph}&\text{\{Rewrite $\ANT{\NT{q}}{\text{Bob}}{\text{Eve}} \ruleto \ANT{\NT{q}}{\text{Bob}}{\text{Dan}}\ \ANT{\NT{q}}{\text{Dan}}{\text{Eve}}$\}}\\
                                           &\ANT{\NT{q}}{\text{Alice}}{\text{Bob}}\ \ANT{\NT{q}}{\text{Bob}}{\text{Dan}}\ \ANT{\NT{q}}{\text{Dan}}{\text{Eve}}\\
    \produces{\Anno\Productions\Graph}&\text{\{Rewrite $\ANT{\NT{q}}{\text{Alice}}{\text{Bob}} \ruleto \textit{friendOf}$, ... \}}\\
                                           &\textit{friendOf}\ \textit{friendOf}\ \textit{friendOf}.
\end{align*}
The node-information in each annotated non-terminal allows us to conclude that there is a path from Alice to Eve of length three, namely the path \[\textbf{Alice}\ \textit{friendOf}\ \textbf{Bob}\ \textit{friendOf}\ \textbf{Dan}\ \textit{friendOf}\ \textbf{Eve}.\]
\end{example}

As Example~\ref{exam:social_network} shows, an annotated non-terminal in an annotated grammar describes the mapping between the trace that is derived by the non-terminal and the first and last node of a path having this trace.

\begin{proposition}
Let $\Grammar = (\NonTerminals, \Alphabet,  \Productions)$ be a context-free grammar, let $\Graph = (\Nodes, \Alphabet, \Transitions)$ be a graph, let $\Anno\Grammar\Graph = (\Anno\NonTerminals\Graph, \Alphabet, \Anno\Productions\Graph)$ be the annotated grammar over $(\Grammar, \Graph)$, let $m\Path n$ be a path in $\Graph$, and let $\NT{a} \in \NonTerminals$ be a non-terminal. We have $\Trace(\Path) \in \Lang(\Grammar; \NT{a})$ if and only if we can derive $\Path$ from $\ANT{\NT{a}}{m}{n} \in \Anno\NonTerminals\Graph$.
\end{proposition}

As annotated grammars are context-free grammars, one can use existing context-free enumeration techniques~\cite{unilex,cfg_uhasselt,dong,lexenum} to produce some of the paths represented by the annotated grammar. The efficiency of these techniques depend on the size of the annotated grammar. We observe the following worst-case upper bounds:

\begin{lemma}\label{lem:graphanno_sizes}
Let $\Grammar = (\NonTerminals, \Alphabet, \Productions)$ be a context-free grammar, let $\Graph = (\Nodes, \Alphabet, \Transitions)$ be a graph, and let $\Anno\Grammar\Graph = (\Anno\NonTerminals\Graph, \Alphabet, \Anno\Productions\Graph)$ be the annotated grammar over $(\Grammar, \Graph)$. We have $\abs{\Anno\NonTerminals\Graph} \leq \abs{\NonTerminals}\abs{\Nodes}^2$ and $\abs{\Anno\Productions\Graph} \leq \abs{\Productions}\abs{\Nodes}^3 + \min(\abs{\NonTerminals}, \abs{\Productions})\abs{\Transitions}$.
\end{lemma}

We propose annotated grammars to represent the query result of a query using the all-path query semantics. Hence, we also need to show how to construct such an annotated grammar, which we do next.

\begin{theorem}\label{thm:anno_grammar}
Let $\Grammar = (\NonTerminals, \Alphabet, \Productions)$ be a context-free grammar, let $\Graph = (\Nodes, \Alphabet, \Transitions)$ be a graph. We can construct the annotated grammar over $(\Grammar, \Graph)$ in $\BigO(\abs{\NonTerminals}\abs{\Transitions} + (\abs{\NonTerminals}\abs{\Nodes})^3)$.
\end{theorem}
\begin{proof}
We use the context-free recognizer for graphs of Hellings~\cite{icdt2014} to construct the set $\Anno\NonTerminals\Graph = \{ \ANT{\NT{a}}{m}{n} \mid \Lang(\Grammar; \NT{a}) \intersect \Lang(\Graph; m, n) \neq \emptyset \}$ in $\BigO(\abs{\NonTerminals}\abs{\Transitions} + (\abs{\NonTerminals}\abs{\Nodes})^3)$. Using $\Anno\NonTerminals\Graph$, we can construct
\begin{align*}
\Anno\Productions\Graph ={}& \{ (\ANT{\NT{a}}{m}{n} \ruleto \symbol) \mid (\NT{a} \ruleto \symbol) \in \Productions \land (m, \symbol, n) \in \Transitions \} \union{}\\
                         &\{ (\ANT{\NT{a}}{m}{n} \ruleto \ANT{\NT{b}}{m}{o}\ \ANT{\NT{c}}{o}{n}) \mid (\NT{a} \ruleto \NT{b}\ \NT{c}) \in \Productions \land \ANT{\NT{a}}{m}{n}, \ANT{\NT{b}}{m}{o}, \ANT{\NT{c}}{o}{n} \in \Anno\NonTerminals\Graph \}.
\end{align*}
For the construction of $\Anno\Productions\Graph$, we represent $\Anno\NonTerminals\Graph$ by a 3-dimensional boolean matrix and the set of production rules $\Productions$ by two look-up structures (one for rules of the form $\NT{a}\ruleto \symbol$ and one for rules of the form $\NT{a} \ruleto \NT{b}\ \NT{c}$). These structures guarantee constant-time lookups for all the parts used in the definition of $\Anno\Productions\Graph$. By the worst-case upper bounds on $\abs{\Anno\Productions\Graph}$, we conclude that we can construct $\Anno\Productions\Graph$ in $\BigO(\abs{\NonTerminals}\abs{\Transitions} + (\abs{\NonTerminals}\abs{\Nodes})^3)$.
\end{proof}

Observe that Theorem~\ref{thm:anno_grammar} also proves Lemma~\ref{lem:cfgintgraph}, and it does so by a direct context-free grammar construction. Usually, Lemma~\ref{lem:cfgintgraph} is proven indirectly by using a pushdown automaton-based construction~\cite{flbook}. Our direct approach to proving Lemma~\ref{lem:cfgintgraph} is essential; indeed, it is the direct construction of a context-free grammar in our proof that allows us to guarantee the structural properties in the result that we need for the derivation of paths.

\section{Answering queries using single-path query semantics}\label{sec:single_path}
Although querying for all paths can be useful in certain cases, it is often sufficient if the query answer contains a single such path: a single path is much easier to comprehend by end users and can already reveal the information end users are looking for. As the length of these paths is not necessarily upper bounded, a logical choice would be to choose a path that is as short as possible. Preferring such a path of minimal length over longer paths can provide additional practical value, as the following example illustrates.

\begin{example}
Recall Example~\ref{exam:social_network}. If Alice used the query $\NT{q}$ to find out how she can get in contact with Eve via friends, friends of friends, and so on, then she probably wants to contact Eve without contacting to many other people. Hence, the provided answer via Bob and Dan is not optimal. The path $\textbf{Alice}\ \textit{friendOf}\ \textbf{Craig}\ \textit{friendOf}\ \textbf{Eve}$ is shorter, and using this path Alice can get in contact with Eva by only contacting Craig.
\end{example}

Towards answering context-free path queries with a single path of minimum length, we proceed in two steps. In Section~\ref{ss:mincfg}, we develop an approach to derive a string of minimum length from a context-free grammar. In Section~\ref{ss:constructpaths}, we apply this approach to derive strings of minimum length to annotated grammars, hence showing how to answer context-free path queries with a path of minimum length.

\subsection{Construction of strings of minimum length}\label{ss:mincfg}
The goal of this section is to provide an approach to finding a string of minimum length in a language defined by a context-free grammar.

\begin{definition}
If $\Lang$ is a language, then the \emph{min-length} of the language $\Lang$, denoted by $\mlanglen{\Lang}$, is defined by $\mlanglen{\Lang} = \min\{ \slen{\String} \mid \String \in \Lang \}$.
\end{definition}

Mclean et al.~\cite{shorttermstring} showed that a string of minimum length in a context-free language can be computed. Their results do, however, not give a practical algorithm or complexity results for deriving strings of minimum length. Towards such a derivation algorithm, we introduce derivations using deterministic non-recursive production rules:
\begin{definition}
Let $\Productions$ be a set of production rules. We define $\heads(\Productions) = \{ \NT{a} \mid (\NT{a} \ruleto \String) \in \Productions \}$.  We define the set of non-terminals derivable from $\NT{a}$ using the production rules in $\Productions$, denoted by $\NTS{\NT{a}}{\Productions}$, as $\NTS{\NT{a}}{\Productions} = \{ \NT{b} \mid \NT{b} \in \NonTerminals \land \exists \String_1\exists \String_2\ \NT{a} \pproduces{\Productions} \String_1 \concat \NT{b} \concat \String_2 \}$. A set of production rules $\Productions$ is \emph{non-recursive} if, for every $\NT{a} \in \heads(\Productions)$, we have $\NT{a} \notin \NTS{\NT{a}}{\Productions}$. A set of production rules $\Productions$ is \emph{deterministic non-recursive} if it is non-recursive; if, for every $\NT{a} \in \heads(\Productions)$, there exists exactly one $(\NT{a} \ruleto \String) \in \Productions$; and if $\NT{a} \in \heads(\Productions)$ implies that there exists a string $\String \in \transitive{\Alphabet}$ such that $\NT{a} \produces{\Productions} \String$.
\end{definition}

Observe that a deterministic non-recursive set $\Productions$ does not provide choices in how one rewrites a non-terminal $\NT{a}$ into a string. As a consequence, $\Productions$ rewrites each $\NT{a} \in \heads(\Productions)$ into a unique string $\String \in \transitive{\Alphabet}$. In this setting, we define $\NDString{\NT{a}}{\Productions} = \String$.

\begin{example}\label{exam:det_nrec}
Recall Example~\ref{exam:social_network}. The following set of production rules in the annotated grammar is deterministic non-recursive:
\begin{align*}
    \ANT{\NT{q}}{\text{A}}{\text{B}} &\ruleto \textit{friendOf}, &
    \ANT{\NT{q}}{\text{A}}{\text{C}} &\ruleto \textit{friendOf},\\
    \ANT{\NT{q}}{\text{B}}{\text{D}} &\ruleto \textit{friendOf},&
    \ANT{\NT{q}}{\text{C}}{\text{E}} &\ruleto \textit{friendOf},\\
    \ANT{\NT{q}}{\text{D}}{\text{E}} &\ruleto \textit{friendOf},&
    \ANT{\NT{q}}{\text{A}}{\text{D}} &\ruleto \ANT{\NT{q}}{\text{A}}{\text{B}}\ \ANT{\NT{q}}{\text{B}}{\text{D}},\\
    \ANT{\NT{q}}{\text{A}}{\text{E}} &\ruleto \ANT{\NT{q}}{\text{A}}{\text{B}}\ \ANT{\NT{q}}{\text{B}}{\text{E}},&
    \ANT{\NT{q}}{\text{B}}{\text{E}} &\ruleto \ANT{\NT{q}}{\text{B}}{\text{D}}\ \ANT{\NT{q}}{\text{D}}{\text{E}}.
\end{align*}
\end{example}

\begin{lemma}\label{lem:unique_deriv_min}
Let $\Grammar = (\NonTerminals, \Alphabet, \Productions)$ be a context-free grammar, and let  $\NT{a} \in \NonTerminals$ be a non-terminal with $\Lang(\Grammar; \NT{a}) \neq \emptyset$. There exists a deterministic non-recursive set $\Productions' \subseteq \Productions$ such that $\slen{\NDString{\NT{a}}{\Productions'}} = \mlanglen{\Lang(\Grammar; \NT{a})}$.
\end{lemma}
\begin{proof}[Proof (sketch).]
Let $\String$ be a string with $\NT{a} \produces{\Productions} \String$ and $\slen{\String} = \mlanglen{\Lang(\Grammar; \NT{a})}$. Consider the derivation of $\String$. If the derivation $\NT{a} \produces{\Productions} \String$ has a sequence of rewrite steps $ \NT{b} \pproduces{\Productions} \String_1 \concat \NT{b} \concat \String_2$, then, due to $\String$ having minimum length, we must have $\String_1 = \String_2 = \EmptyString$. Hence, we can remove all the rewrite steps involved in $\NT{b} \pproduces{\Productions} \String_1 \concat \NT{b} \concat \String_2$  and use the rewrite steps used to rewrite $\NT{b}$ in $\String_1 \concat \NT{b} \concat \String_2$ instead. If the derivation $\NT{a} \produces{\Productions} \String$ uses distinct production rules $\NT{c} \ruleto \String_1$ and $\NT{c} \ruleto \String_2$, then, due to $\String$ having minimum length, both $\String_1$ and $\String_2$ are rewritten into equal length strings in $\transitive{\Alphabet}$. Hence, we can choose one of the two production rules and use it for both rewrites of $\NT{c}$, the resulting string will have the same length as $\String$.
\end{proof}

\begin{corollary}\label{cor:unique_deriv_min}
Let $\Grammar = (\NonTerminals, \Alphabet, \Productions)$ be a context-free grammar. There exists a deterministic non-recursive set $\Productions' \subseteq \Productions$ such that for every non-terminal $\NT{a} \in \NonTerminals$ with $\Lang(\Grammar; \NT{a}) \neq \emptyset$, we have $\slen{\NDString{\NT{a}}{\Productions'}} = \mlanglen{\Lang(\Grammar; \NT{a})}$.
\end{corollary}

We say that a deterministic non-recursive set satisfying the conditions of Corollary~\ref{cor:unique_deriv_min} is \emph{minimizing}. Corollary~\ref{cor:unique_deriv_min} does not imply that each deterministic non-recursive set always produces strings of minimum length. With an example, we show that this is not the case:

\begin{example}\label{exam:det_nrec_min}
Recall Example~\ref{exam:det_nrec}.  The provided set of deterministic non-recursive production rules $\Productions'$ is not minimizing. We have $\slen{\NDString{\ANT{\NT{q}}{\text{A}}{\text{E}}}{ \Productions'}} = 3$, while a shorter string of length two exists (via Craig). By replacing the production rule for $\ANT{\NT{q}}{\text{Alice}}{\text{Eve}}$ by $\ANT{\NT{q}}{\text{Alice}}{\text{Eve}} \ruleto \ANT{\NT{q}}{\text{Alice}}{\text{Craig}}\ \ANT{\NT{q}}{\text{Craig}}{\text{Eve}}$, the resulting set of production rules is minimizing.
\end{example}

Given a minimizing set of production rules $\Productions'$, it is straightforward to produce a string of minimum length for each non-terminal $\NT{a} \in \heads(\Productions')$ that has such a string. The worst-case complexity of producing these strings is dominated by the length of the produced strings. By using the restrictions put on deterministic non-recursive sets, we can provide the following worst-case upper bounds on the length of strings of minimal length:

\begin{proposition}\label{prop:cfg_minlen_upperbound}
Let $\Grammar = (\NonTerminals, \Alphabet, \Productions)$ be a context-free grammar and let $N = \{ \NT{a} \mid \NT{a} \in \NonTerminals \land \Lang(\Grammar; \NT{a}) \neq \emptyset \}$ be the set of non-terminals that define a non-empty language.  We have $\max_{\NT{a} \in N} (\mlanglen{\Lang(\Grammar; \NT{a})}) \leq 2^{\abs{\NonTerminals} - 1}$ and $\sum_{\NT{a} \in N} (\mlanglen{\Lang(\Grammar; \NT{a})}) \leq 2^{\abs{\NonTerminals}} - 1$.
\end{proposition}
\begin{proof}[Proof (sketch).]
Let $\abs{\NonTerminals} = i$ and let $\symbol \in \Alphabet$. We only have to consider the case where $\Productions$ is a deterministic non-recursive set. Hence, we can order the non-terminals such that $\NonTerminals = \{ \NT{a}_0, \dots, \NT{a}_{i-1} \}$ and we use the production rules $\NT{a}_0 \ruleto \symbol$ and $\NT{a}_j \ruleto \NT{a}_{j-1}\ \NT{a}_{j-1}$, for all $1 \leq j \leq i-1$.
\end{proof}

As there exists a straightforward procedure to efficiently construct strings of minimum length from a minimizing set of production rules, we only need a procedure to construct such a set of production rules for a given a context-free grammar. Algorithm~\ref{alg:detnonrec} provides such a procedure.

\begin{algorithm}[ht!]
\caption{Construct a minimizing set of production rules for $\Grammar = (\NonTerminals, \Alphabet, \Productions)$}\label{alg:detnonrec}
\begin{algorithmic}[1]
\STATE $\Productions', \VAR{cost} \GETS \text{empty mapping}, \text{empty mapping}$
\STATE $\VAR{new}$ is a min-priority queue
\FORALL{$(\NT{a} \ruleto \symbol) \in \Productions$}\label{alg:detnonrec:init}
    \IF{$\NT{a} \notin \VAR{cost}$}
        \STATE $\VAR{cost}[\NT{a}], \Productions'[\NT{a}] \GETS 1, (\NT{a} \ruleto \symbol)$
        \STATE add $\NT{a}$ to $\VAR{new}$ with priority $1$
    \ENDIF
\ENDFOR
\WHILE{$\VAR{new} \neq \emptyset$}\label{alg:detnonrec:main}
    \STATE take $\NT{a}$ with minimum priority in $\VAR{new}$, and remove it from $\VAR{new}$
    \FORALL{$(\NT{c} \ruleto \NT{a}\ \NT{b}) \in \Productions$ with $\NT{b} \in \VAR{cost}$}\label{alg:detnonrec:first}
        \STATE $\textsc{produce}(\NT{c} \ruleto \NT{a}\ \NT{b})$
    \ENDFOR
    \FORALL{$(\NT{c} \ruleto \NT{b}\ \NT{a}) \in \Productions$ with $\NT{b} \in \VAR{cost}$}\label{alg:detnonrec:second}
        \STATE $\textsc{produce}(\NT{c} \ruleto \NT{b}\ \NT{a})$
    \ENDFOR
\ENDWHILE
\RETURN $\{ \Productions'[\NT{a}] \mid \NT{a} \in \Productions'\}$
\PROCEDURE{$\textsc{produce}(\NT{d} \ruleto \NT{e}\ \NT{f})$}
    \IF{$\NT{d} \notin \VAR{cost}$}
        \STATE $\VAR{cost}[\NT{d}], \Productions'[\NT{d}] \GETS  \VAR{cost}[\NT{e}] + \VAR{cost}[\NT{f}], (\NT{d} \ruleto \NT{e}\ \NT{f})$
        \STATE add $\NT{d}$ to $\VAR{new}$ with priority $\VAR{cost}[\NT{e}] + \VAR{cost}[\NT{f}]$
    \ELSIF{$\VAR{cost}[\NT{d}] > \VAR{cost}[\NT{e}] + \VAR{cost}[\NT{f}]$}
        \STATE $\VAR{cost}[\NT{d}], \Productions'[\NT{d}] \GETS \VAR{cost}[\NT{e}] + \VAR{cost}[\NT{f}], (\NT{d} \ruleto \NT{e}\ \NT{f})$
        \STATE lower priority of $\NT{d}$ in $\VAR{new}$ to $\VAR{cost}[\NT{e}] + \VAR{cost}[\NT{f}]$
    \ENDIF
\end{algorithmic}
\end{algorithm}

\begin{proposition}\label{prop:detnonrec}
Let $\Grammar = (\NonTerminals, \Alphabet, \Productions)$ be a context-free grammar. Algorithm~\ref{alg:detnonrec} applied on $\Grammar$ produces a minimizing set of production rules for $\Grammar$. 
\end{proposition}
\begin{proof}[Proof (sketch)]
The main \emph{while}-loop maintains the following invariants:
\begin{enumerate}
\item \label{prop:detnonrec:complete} If $\NT{a} \in \Productions'$ and $\Productions'[\NT{a}] = (\NT{a} \ruleto \NT{b}\ \NT{c})$, then $\NT{b}, \NT{c} \in \Productions'$.

\item \label{prop:detnonrec:costfun} If $\NT{a} \in \Productions'$ and $\Productions'[\NT{a}] = (\NT{a} \ruleto \NT{b}\ \NT{c})$, then  $\VAR{cost}[\NT{a}] \geq \VAR{cost}[\NT{b}] + \VAR{cost}[\NT{c}]$, $\VAR{cost}[\NT{a}] > \VAR{cost}[\NT{b}]$, and $\VAR{cost}[\NT{a}] > \VAR{cost}[\NT{c}]$.

\item \label{prop:detnonrec:coststr} If $\NT{a} \in \Productions'$, then $\slen{\NDString{\NT{a}}{S}} \leq \VAR{cost}[\NT{a}]$.

\item \label{prop:detnonrec:orderins} Let $m$ be the priority of the last element removed from $\VAR{new}$. No new element is inserted in $\VAR{new}$ with priority less than or equal to $m$.

\item \label{prop:detnonrec:minlength}  Let $m$ be the priority of the last element removed from $\VAR{new}$. For all $\NT{a} \in \NonTerminals$ with $\mlanglen{\Lang(\Grammar; \NT{a})} \leq m$, we have $\VAR{cost}[\NT{a}] = \mlanglen{\Lang(\Grammar; \NT{a})}$.
\end{enumerate}
As each non-terminal is added to $\VAR{new}$ at most once, Algorithm~\ref{alg:detnonrec} terminates. At termination, Invariants~\ref{prop:detnonrec:complete}--\ref{prop:detnonrec:minlength} guarantee that the resulting set of production rules $\{ \Productions'[\NT{a}] \mid \NT{a} \in \Productions'\}$ is minimizing.
\end{proof}

We observe that we cannot straightforwardly generalize Algorithm~\ref{alg:detnonrec} to the setting that includes production rules of the form $\NT{a} \ruleto \EmptyString$: Invariants~\ref{prop:detnonrec:costfun},~\ref{prop:detnonrec:orderins}, and~\ref{prop:detnonrec:minlength} of the proof of Proposition~\ref{prop:detnonrec} no longer hold (as they all require a strict ordering of the $\VAR{cost}$ of non-terminals). This issue can be resolved by maintaining a timestamp on each non-terminal (such that a non-terminal has timestamp $i$ if it was the $i$-th change to $\Productions'$), and change the relevant invariants to not require a strict ordering on the $\VAR{cost}$ of non-terminals, but a strict ordering on the pairs $(\VAR{cost}, \text{timestamp})$ of non-terminals.

\begin{theorem}\label{thm:detnonrec}
Let $\Grammar = (\NonTerminals, \Alphabet, \Productions)$ be a context-free grammar. Algorithm~\ref{alg:detnonrec} constructs a minimizing set of production rules for $\Grammar$ in $\BigO(\abs{\NonTerminals}(\abs{\NonTerminals} \log(\abs{\NonTerminals}) + \abs{\Productions}))$.
\end{theorem}
\begin{proof}
We represent $\VAR{costs}$ as an array holding $\abs{\NonTerminals}$ integers. The costs used in $\VAR{cost}$ and $\VAR{new}$ are integers in the range $1, \dots, 2^{\abs{\NonTerminals} - 1}$. We can represent each of these integers using $\log(2^{\abs{\NonTerminals}}) = \abs{\NonTerminals}$ bits.  The initialization steps perform $\BigO(\abs{\Productions})$ steps. The \emph{while}-loop will, in the worst case, visit every non-terminal once. For each of these non-terminals, one insertion into and one removal from the priority queue $\VAR{new}$ is performed. The inner \emph{for}-loops will visit every production rule twice, causing at most $2\abs{\Productions}$ decrease key operations on priority queue $\VAR{new}$. When using a Fibonacci heap for a priority queue holding at most $e$ elements, each insert and removal costs $\BigO(\log(e))$ and each decrease key operation costs an amortized $\BigO(1)$ heap operations~\cite{intro_algo,fibheap}. Hence, a total of $\BigO(\abs{\NonTerminals} \log(\abs{\NonTerminals}) + \abs{\Productions})$ heap operations are performed. Taking the size of the integers representing priorities into account, the heap operations cost $\BigO(\abs{\NonTerminals}(\abs{\NonTerminals} \log(\abs{\NonTerminals}) + \abs{\Productions}))$.
\end{proof}

Using Theorem~\ref{thm:detnonrec} and Proposition~\ref{prop:cfg_minlen_upperbound}, we conclude the following:

\begin{corollary}
Let $\Grammar = (\NonTerminals, \Alphabet, \Productions)$ be a context-free grammar, let $N = \{ \NT{a} \mid \NT{a} \in \NonTerminals \land \Lang(\Grammar; \NT{a}) \neq \emptyset \}$ be the set of non-terminals that define a non-empty language, and let $L = \sum_{\NT{a} \in N} \mlanglen{\Lang(\Grammar, \NT{a})}$ be the combined length of a string of minimum length for each non-terminal in $N$. We can construct strings of minimum length for all non-terminals in $N$ in $\BigO(\abs{\NonTerminals}(\abs{\NonTerminals}\log(\abs{\NonTerminals}) + \abs{\Productions}) + L) = \BigO(\abs{\NonTerminals}(\abs{\NonTerminals}\log(\abs{\NonTerminals}) + \abs{\Productions}) + 2^{\abs{\NonTerminals}})$.
\end{corollary}

\subsection{Construction of paths of minimum length}\label{ss:constructpaths}
We can already answer queries with single paths of minimum length by first constructing an annotated grammar and then applying Algorithm~\ref{alg:detnonrec}. This approach  has high overhead due to the explicit construction and storing of the annotated grammar. To reduce this overhead, we adapt Algorithm~\ref{alg:detnonrec} to the setting of query evaluation using the single-path query semantics. The resulting algorithm, Algorithm~\ref{alg:detnonrecag}, operates on a normal context-free grammar and a graph, and  derives the necessary details of the annotated grammar in place. If necessary, Algorithm~\ref{alg:detnonrecag} can use straightforward bookkeeping to also construct  $\Anno\NonTerminals\Graph$ and $\Anno\Productions\Graph$, this without increasing the asymptotic complexity of the algorithm.

\begin{algorithm}[ht!]
\caption{Construct a minimizing set of production rules for the annotated grammar over $\Grammar = (\NonTerminals, \Alphabet, \Productions)$ and $\Graph = (\Nodes, \Alphabet, \Transitions)$}\label{alg:detnonrecag}
\begin{algorithmic}[1]
\STATE $\Productions', \VAR{cost} \GETS \text{empty mapping}, \text{empty mapping}$
\STATE $\VAR{new}$ is a min-priority queue
\FORALL{$(\NT{a} \ruleto \symbol) \in \Productions$ and $(m, \symbol, n) \in \Transitions$}\label{alg:detnonrecag:init}
    \IF{$\ANT{\NT{a}}{m}{n} \notin \VAR{cost}$}
        \STATE $\VAR{cost}[\ANT{\NT{a}}{m}{n}], \Productions'[\ANT{\NT{a}}{m}{n}] \GETS 1, (\ANT{\NT{a}}{m}{n} \ruleto \symbol)$
        \STATE add $\ANT{\NT{a}}{m}{n}$ to $\VAR{new}$ with priority $1$
    \ENDIF
\ENDFOR
\WHILE{$\VAR{new} \neq \emptyset$}\label{alg:detnonrecag:main}
    \STATE take $\ANT{\NT{a}}{m}{n}$ with minimum priority in $\VAR{new}$, and remove it from $\VAR{new}$
    \FORALL{$(\NT{c} \ruleto \NT{a}\ \NT{b}) \in \Productions$ with $\ANT{\NT{b}}{n}{o} \in \VAR{cost}$}\label{alg:detnonrecag:first}
        \STATE $\textsc{produce}(\ANT{\NT{c}}{m}{o} \ruleto \ANT{\NT{a}}{m}{n}\ \ANT{\NT{b}}{n}{o})$
    \ENDFOR
    \FORALL{$(\NT{c} \ruleto \NT{b}\ \NT{a}) \in \Productions$ with $\ANT{\NT{b}}{o}{m} \in \VAR{cost}$}\label{alg:detnonrecag:second}
        \STATE $\textsc{produce}(\ANT{\NT{c}}{o}{n} \ruleto \ANT{\NT{b}}{o}{m}\ \ANT{\NT{a}}{m}{n})$
    \ENDFOR
\ENDWHILE
\RETURN $\{ \Productions'[\ANT{\NT{a}}{m}{n}] \mid \ANT{\NT{a}}{m}{n} \in \Productions'\}$
\PROCEDURE{$\textsc{produce}(\ANT{\NT{d}}{u}{w} \ruleto \ANT{\NT{e}}{u}{v}\ \ANT{\NT{f}}{v}{w})$}
    \IF{$\ANT{\NT{d}}{u}{w} \notin \VAR{cost}$}
        \STATE $\VAR{cost}[\ANT{\NT{d}}{u}{w}], \Productions'[\ANT{\NT{d}}{u}{w}] \GETS  \VAR{cost}[\ANT{\NT{e}}{u}{v}] + \VAR{cost}[\ANT{\NT{f}}{v}{w}], (\ANT{\NT{d}}{u}{w} \ruleto \ANT{\NT{e}}{u}{v}\ \ANT{\NT{f}}{v}{w})$
        \STATE add $\ANT{\NT{d}}{u}{w}$ to $\VAR{new}$ with priority $\VAR{cost}[\ANT{\NT{e}}{u}{v}] + \VAR{cost}[\ANT{\NT{f}}{v}{w}]$
    \ELSIF{$\VAR{cost}[\ANT{\NT{d}}{u}{w}] > \VAR{cost}[\ANT{\NT{e}}{u}{v}] + \VAR{cost}[\ANT{\NT{f}}{v}{w}]$}
        \STATE $\VAR{cost}[\ANT{\NT{d}}{u}{w}], \Productions'[\ANT{\NT{d}}{u}{w}] \GETS \VAR{cost}[\ANT{\NT{e}}{u}{v}] + \VAR{cost}[\ANT{\NT{f}}{v}{w}], (\ANT{\NT{d}}{u}{w} \ruleto \ANT{\NT{e}}{u}{v}\ \ANT{\NT{f}}{v}{w})$
        \STATE lower priority of $\ANT{\NT{d}}{u}{w}$ in $\VAR{new}$ to $\VAR{cost}[\ANT{\NT{e}}{u}{v}] + \VAR{cost}[\ANT{\NT{f}}{v}{w}]$
    \ENDIF
\end{algorithmic}
\end{algorithm}

Before we fully analyze Algorithm~\ref{alg:detnonrecag}, we use Lemma~\ref{lem:graphanno_sizes} and Proposition~\ref{prop:cfg_minlen_upperbound} to conclude the following naive worst-case upper bound on the length of paths of minimal length:

\begin{corollary}\label{cor:nav_cfg_minlen_cfgintfa}
Let $\Grammar = (\NonTerminals, \Alphabet, \Productions)$ be a context-free grammar with $\NT{a} \in \NonTerminals$ and let $\Graph = (\Nodes, \Alphabet, \Transitions)$ be a graph with $m, n \in \Nodes$, such that $\Lang = \Lang(\Grammar; \NT{a}) \intersect \Lang(\Graph; m, n) \neq \emptyset$. We have $\mlanglen{\Lang} \leq 2^{\abs{\NonTerminals}\abs{\Nodes}^2 - 1}$.
\end{corollary}

Using Corollary~\ref{cor:nav_cfg_minlen_cfgintfa}, we conclude the following:

\begin{proposition}
Let $\Grammar = (\NonTerminals, \Alphabet, \Productions)$ be a context-free grammar and let $\Graph = (\Nodes, \Alphabet, \Transitions)$ be a graph. Algorithm~\ref{alg:detnonrecag} constructs a minimizing set of production rules for the annotated grammar $\Anno\Grammar\Graph = (\Anno\NonTerminals\Graph, \Alphabet, \Anno\Productions\Graph)$ over $(\Grammar, \Graph)$ in $
\BigO(\abs{\Anno\NonTerminals\Graph}(\abs{\Anno\NonTerminals\Graph} \log(\abs{\Anno\NonTerminals\Graph}) + \abs{\Anno\Productions\Graph}))$, and, hence, in $\BigO(\abs{\NonTerminals}\abs{\Nodes}^2((\abs{\NonTerminals}\abs{\Nodes}^2)\log(\abs{\NonTerminals}\abs{\Nodes}^2) + \abs{\Productions}\abs{\Nodes}^3 + \min(\abs{\NonTerminals},\abs{\Productions})\abs{\Transitions}))$.
\end{proposition}

\begin{corollary}
Let $\Grammar = (\NonTerminals, \Alphabet, \Productions)$ be a context-free grammar with $\NT{a} \in \NonTerminals$ and let $\Graph = (\Nodes, \Alphabet, \Transitions)$ be a graph with $m, n \in \Nodes$, such that $\Lang = \Lang(\Grammar; \NT{a}) \intersect \Lang(\Graph; m, n) \neq \emptyset$. We can construct a path $m\Path n$ such that $\Trace(\Path) \in \Lang$ and $\slen{\Trace(\Path)} = \mlanglen{\Lang}$ in \begin{multline*}\BigO(\abs{\NonTerminals}\abs{\Nodes}^2((\abs{\NonTerminals}\abs{\Nodes}^2)\log(\abs{\NonTerminals}\abs{\Nodes}^2) + \abs{\Productions}\abs{\Nodes}^3 + \min(\abs{\NonTerminals},\abs{\Productions})\abs{\Transitions}) + \slen{\Trace(\Path)}) ={}\\ \BigO(\abs{\NonTerminals}\abs{\Nodes}^2((\abs{\NonTerminals}\abs{\Nodes}^2)\log(\abs{\NonTerminals}\abs{\Nodes}^2) + \abs{\Productions}\abs{\Nodes}^3 + \min(\abs{\NonTerminals},\abs{\Productions})\abs{\Transitions}) + 2^{\abs{\NonTerminals}\abs{\Nodes}^2 - 1}).\end{multline*}
\end{corollary}

Observe that the upper bound of Corollary~\ref{cor:nav_cfg_minlen_cfgintfa} is very loose: Proposition~\ref{prop:cfg_minlen_upperbound} depends on production rules of the form $\NT{a} \rightarrow \NT{b}\ \NT{b}$, whereas, in general, annotated grammars only allow for such production rules in very restricted cases. Hence, we look at ways to improve the worst-case upper bound observed by Corollary~\ref{cor:nav_cfg_minlen_cfgintfa}. As an initial step, we consider languages defined over singleton alphabets:

\begin{proposition}\label{prop:un_cfg_minlen_cfgintfa}
Let $\Alphabet$ be an alphabet with $\abs{\Alphabet} = 1$, let $\Grammar = (\NonTerminals, \Alphabet, \Productions)$ be a context-free grammar with $\NT{a} \in \NonTerminals$, and let $\Graph = (\Nodes, \Alphabet, \Transitions)$ be a graph with $m, n \in \Nodes$, such that $\Lang = \Lang(\Grammar; \NT{a}) \intersect \Lang(\Graph; m, n) \neq \emptyset$. In the worst case, we have $\abs{\Nodes}2^{\abs{\NonTerminals} - 1} \leq \mlanglen{\Lang} \leq \abs{\Nodes}(2^{2\abs{\NonTerminals} - 1} + 1)$.
\end{proposition}
\begin{proof}
First, we prove the lower bound. Let $\Alphabet = \{ \symbol \}$, let $\Nodes = \{ n_0, \dots, n_{\abs{\Nodes}-1} \}$, let $\Transitions = \{ (n_i, \symbol, n_{i+1 \mod \abs{\Nodes}}) \mid 0 \leq i \leq \abs{\Nodes} - 1 \}$, let $\NonTerminals = \{ \NT{a}_0, \dots, \NT{a}_{\abs{\NonTerminals}-1} \}$, and let $\Productions = \{ \NT{a}_0 \ruleto \symbol, \NT{a}_0 \ruleto \NT{a}_0\ \NT{a}_0 \} \union \{ \NT{a}_i \ruleto \NT{a}_{i-1}\ \NT{a}_{i-1} \mid 1 \leq i \leq \abs{\NonTerminals} - 1 \}$.

With these definitions we have $\Lang(\Graph; n_i, n_i) = \{ \symbol^{k\abs{\Nodes}} \mid 0 \leq k \}$, for every $1 \leq i \leq \abs{\Nodes}$, and we have $\Lang(\Grammar; \NT{a}_0) = \{ \symbol^k \mid 1 \leq k \}$. Hence, the string $\String'$ of minimum length such that $\ANT{\NT{a}_0}{n_i}{n_i} \produces{\Productions} \String'$ is $\String' = \symbol^\abs{\Nodes}$. Each $\NT{a}_j$, $1 \leq j \leq \abs{\NonTerminals} - 1$, will be rewritten in a string of exactly $2^j$ $\NT{a}_0$ non-terminals. Hence, the string $\String'_j$ of minimum length such that $\ANT{\NT{a}_j}{n_i}{n_i} \produces{\Productions} \String'_j$ is $\String'_j =  \smash{\symbol^{\abs{\Nodes}2^j}}$, and we conclude $\slen{\String'_{\abs{\NonTerminals} - 1}} = \abs{\Nodes}2^{\abs{\NonTerminals} - 1}$.

The upper bound is proven using a result of Pighizzini et al.~\cite{pighizzini}: for each context-free grammar $\Grammar$ with $y$ non-terminals, there exists a finite automaton $\Graph' = (\Nodes', \Alphabet, \Transitions')$ with initial state $m$, final states $F$, and with $\abs{\Nodes'} \leq 2^{2\abs{\NonTerminals} - 1} + 1$ such that $(\bigcup_{n \in F} \Lang(\Graph; m, n)) = \Lang(\Grammar; \NT{a})$. We use $\Graph'$ to represent $\Grammar$ and we apply the well-known product construction for the intersection of finite automata on $\Graph'$ and $\Graph$. The resulting finite automaton has $\abs{\Nodes}\abs{\Nodes'} = \abs{\Nodes}(2^{2\abs{\NonTerminals} - 1} + 1)$ states, proving the upper bound.
\end{proof}

Due to Proposition~\ref{prop:un_cfg_minlen_cfgintfa}, we can conclude that in the case of unlabeled graphs, the complexity of query evaluation with the single-path query semantics is polynomial in terms of the graph size and exponential in terms of the query size. Observe, however, that the exponential complexity in terms of the query size follows straightforward from the succinctness of context-free grammars (as compared to regular expressions and finite automata).

In the labeled case, we can still use Proposition~\ref{prop:un_cfg_minlen_cfgintfa} to get worst-case lower bounds on $\mlanglen{\Lang(\Grammar; \NT{a}) \intersect \Lang(\Graph; m, n)}$. The worst-case upper bound provided by Proposition~\ref{prop:un_cfg_minlen_cfgintfa} can, however, not be generalized to arbitrary alphabets, which we show next.

\begin{proposition}\label{prop:cfg_minlen_cfgintfa}
Let $\Grammar = (\NonTerminals, \Alphabet, \Productions)$ be a context-free grammar with $\NT{a} \in \NonTerminals$, and let $\Graph = (\Nodes, \Alphabet, \Transitions)$ be a graph with $m, n \in \Nodes$, such that $\Lang = \Lang(\Grammar; \NT{a}) \intersect \Lang(\Graph; m, n) \neq \emptyset$. In the worst case, we have $\abs{\Nodes}^2 2^{\abs{\NonTerminals}} \leq 64\mlanglen{\Lang}$.
\end{proposition}
\begin{proof}
Choose the well-known context-free language $\Lang = \{ {\symbol_1}^k{\symbol_2}^k \mid 1 \leq k \}$. The context-free grammar $\Grammar = (\NonTerminals, \Alphabet, \Productions)$ with $\NonTerminals = \{ \NT{a}, \NT{a}', \NT{a}_1, \NT{a}_2, \NT{b}_1, \dots, \NT{b}_{\abs{\NonTerminals} - 4} \}$ and
\begin{align*}
\Productions ={}& \{ \NT{a} \ruleto \NT{a}_1\ \NT{a}_2, \NT{a} \ruleto \NT{a}_1\ \NT{a}', \NT{a}' \ruleto \NT{a}\ \NT{a}_2, \NT{a}_1 \ruleto \symbol_1, \NT{a}_2 \ruleto \symbol_2 \} \union{}\\
                &\{  \NT{b}_1 \ruleto \NT{a}\ \NT{a} \} \union \{ \NT{b}_j \ruleto \NT{b}_{j-1}\ \NT{b}_{j-1} \mid 1 < j \leq \abs{\NonTerminals} - 4 \}
\end{align*}
has $\Lang(\Grammar; \NT{a}) = \Lang$. Choose a $k$ with $1 \leq k$ and choose $\abs{Q} = u + v - 1$ with $u = 2^k + 1$ and $v = u - 1$. Let $\Graph = (\Nodes, \Alphabet, \Transitions)$ with $\Nodes = \{c, m_1, \dots, m_{u-1}, n_1, \dots n_{v-1} \}$ and
\begin{align*}
        \Transitions ={}& \{ (c, \symbol_1, m_1), (m_{u-1}, \symbol_1, c) \} \union \{ (m_i, \symbol_1, m_{i+1}) \mid 1 \leq i < u - 1 \} \union{}\\
                       &  \{ (c, \symbol_2, n_1), (n_{v - 1}, \symbol_2, c) \}  \union \{ (n_i, \symbol_2, n_{i+1}) \mid 1 \leq i < v - 1 \}.
\end{align*}
The resulting graph is visualized in Figure~\ref{fig:twocyclesparen}.
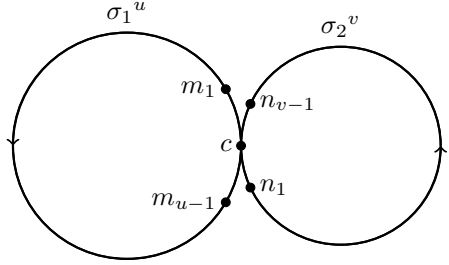
\begin{figure}[t!]
    \centering
    \begin{tikzpicture}[dot/.style={circle,scale=0.35,draw=black,fill=black},scale=0.75]
        \coordinate (l) at (2, 0);
        \coordinate (r) at (5.75, 0);
        
        \draw[->,thick] (0, 0) arc[radius=-2, start angle=0, end angle=360];
        \draw[->,thick] (7.5, 0) arc[radius=1.75, start angle=0, end angle=360];
        
        \draw[thick] (l) circle [radius=2];
        \draw[thick] (r) circle [radius=1.75];

        \node[above] at (2, 2) {${\symbol_1}^u$};
        \node[above] at (5.75, 1.75) {${\symbol_2}^v$};
        
         \node[dot] (mf) at ($(l)!1!30:(4, 0)$) {};
         \node[dot] (ml) at ($(l)!1!330:(4, 0)$) {};
         
         \node[dot] (nl) at ($(r)!1!155:(7.5, 0)$) {};
         \node[dot] (nf) at ($(r)!1!205:(7.5, 0)$) {};
        
        \node[dot] (c) at (4, 0) {};
        \node[left] at (c) {$c$};
        
        \node[left] at (mf) {$m_1$};
        \node[left] at (ml) {$m_{u-1}$};
        
        \node[right] at (nf) {$n_1$};
        \node[right] at (nl) {$n_{v-1}$};
    \end{tikzpicture}
    \caption{The double-cyclic graph: two cycles, one having $u$ edges labeled with $\symbol_1$, and one having $v$ edges labeled with $\symbol_2$. The two cycles are connected via a shared node $c$.}\label{fig:twocyclesparen}
\end{figure}

Let $c \Path c$ be a path in $\Graph$ with $\Trace(\Path) \in \Lang(\Grammar; \NT{a})$. Due to the definition of $\NT{a}$, we must have $\Trace(\Path) = \String_1 \concat \String_2$ with $\String_1 = {\symbol_1}^x$ and $\String_2 = {\symbol_2}^x$, for $1 \leq x$. Due to the structure of the graph, $\String_1$ must be the trace of a path $c\Path_1c$ and $\String_2$ must be the trace of a path $c\Path_2c$ in graph $\Graph$. From these constraints, we conclude $\Lang(\Grammar; \NT{a}) \intersect \Lang(\Graph; c, c) = \{ {\symbol_1}^{k\lcm(u, v)}{\symbol_2}^{k\lcm(u, v)} \mid 1 \leq k \}$. Observe that $u = 2^k$ and $v = 2^k+1$ are coprime, hence, we have $\lcm(u, v) = uv$.

Each $\NT{b}_j$, $1 \leq j \leq \abs{\NonTerminals} - 4$, will be rewritten in a string of exactly $2^j$ $\NT{a}$ non-terminals. As only node $c$ has outgoing edges labeled with both $\symbol_1$ and $\symbol_2$, each $\ANT{\NT{b}_j}{c}{c}$ will be rewritten in a string of exactly $2^j$ $\ANT{\NT{a}}{c}{c}$ non-terminals. Hence, we conclude $\mlanglen{\Lang(\Grammar; \NT{b}_j) \intersect \Lang(\Graph; c, c)} = uv2^j$, and we conclude \[\mlanglen{\Lang(\Grammar; \NT{b}_{\abs{\NonTerminals} - 4}) \intersect \Lang(\Graph; c, c)} = uv2^{\abs{\NonTerminals}-4} > v^2 \frac{2^{\abs{\NonTerminals}}}{16} = \left(\frac{\abs{\Nodes}}{2}\right)^2 \frac{2^{\abs{\NonTerminals}}}{16} = \frac{\abs{\Nodes}^2 2^{\abs{\NonTerminals}}}{64}.\qedhere\]
\end{proof}

For labeled graphs, we do not yet have a better worst-case upper bound than the naive upper-bound provided by Corollary~\ref{cor:nav_cfg_minlen_cfgintfa}.
\begin{problem}\label{prob:inter_re_cfg}
Let $\Grammar = (\NonTerminals, \Alphabet, \Productions)$ be a context-free grammar with $\NT{a} \in \NonTerminals$, and let $\Graph = (\Nodes, \Alphabet, \Transitions)$ be a graph with $m, n \in \Nodes$, such that $\Lang = \Lang(\Grammar; \NT{a}) \intersect \Lang(\Graph; m, n) \neq \emptyset$. What is the strict worst-case upper bound on $\mlanglen{\Lang}$? Or, equivalently, what is the strict worst-case upper bound on the length of a shortest string in the intersection of the language of a context-free grammar with $x$ non-terminals and the language of a finite automaton with $y$ states?
\end{problem}
 We conjecture that, as in the unlabeled case, the complexity of query evaluation with the single-path query semantics is polynomial in terms of the graph size and exponential in terms of the query size.

\section{Experimental results}\label{sec:exp}

To provide insight in the practical behavior of path-based query evaluation, we have implemented algorithms for the evaluation of queries using the single-path query semantics.\footnote{The algorithms are implemented in $\texttt{C++}$. Measurements where performed on a system with an Intel Core i5-4670 CPU, running at a maximum of 3.8GHz, and with 16GB of main memory. The source code will be made available under an open-source license.} We primarily focus on the running time of Algorithm~\ref{alg:detnonrecag}, as the cost of producing the paths of interest heavily depends on whether one wants to produce a path for a particular node pair or for all node pairs. We perform three different tests:
\begin{enumerate}
\item We compare two context-free grammars that both evaluate to the positive transitive closure (under the relational query semantics):
\begin{align*}
\NT{q}_1 &\ruleto \NT{a}\ \NT{q}_1,& \NT{q}_1 &\ruleto \symbol,& \NT{a}   &\ruleto \symbol;\\
\NT{q}_2 &\ruleto \NT{q}_2\ \NT{q}_2,&   \NT{q}_2 &\ruleto \symbol.
\end{align*}
Observe that the context-free grammar $\NT{q}_1$ is linear and non-ambiguous, whereas the context-free grammar with non-terminal $\NT{q}_2$ is non-linear and highly ambiguous. We measure the running time of constructing minimizing sets of annotated production rules for these queries on the cyclic graphs of Proposition~\ref{prop:un_cfg_minlen_cfgintfa}.
 \item We compare the context-free grammar $\NT{q}_1$, which produces dense result sets, with the language $\Lang = \{\symbol\symbol\symbol \}$, which produces sparse result sets. We measure the running time of constructing minimizing sets of annotated production rules for these queries on the cyclic graphs of Proposition~\ref{prop:un_cfg_minlen_cfgintfa}.
\item We derive the longest path of minimum length that matches the following context-free grammar:
\begin{align*}
\NT{q}  &\ruleto \NT{a}\ \NT{q}',   &\NT{q}' &\ruleto \NT{q}\ \NT{b},   &\NT{q}  &\ruleto \NT{a}\ \NT{b},\\
\NT{a} &\ruleto \symbol_1,          &\NT{b} &\ruleto \symbol_2.
\end{align*}
Similar context-free grammars are used in Example~\ref{exam:same_generations} and Proposition~\ref{prop:cfg_minlen_cfgintfa}. We evaluate these queries on the double-cyclic graphs of Proposition~\ref{prop:cfg_minlen_cfgintfa}, for which we know that the query $\NT{q}$ will produce paths with a high minimum length. We measure both the running time of constructing a minimizing set of annotated production rules for $\NT{q}$ on the double-cyclic graphs, and the running time for deriving the longest path of minimum length from the resulting minimizing set.
\end{enumerate}

We remark that these tests illustrate extreme behavior: the queries apply to the entirety of the graph  and all nodes and edges will participate in the outcome. This does not reflect all practical applications, where one can often expect that queries are much more selective.

%
%
%
%
%
\pgfplotstableread{
    nodes	edges	non-terminals	output	max	paths	minimize-time	produce-time
    125	125	1	984375	125	15625	61075875	12981
    375	375	1	26437500	375	140625	1552183901	62488
    625	625	1	122265625	625	390625	7633070200	110183
    875	875	1	335343750	875	765625	24687993600	252970
    1125	1125	1	712546875	1125	1265625	56340109882	406019
    1375	1375	1	1300750000	1375	1890625	106192905378	594690
}\dataTcNLCycle

\pgfplotstableread{
    nodes	edges	non-terminals	output	max	paths	minimize-time	produce-time
    250	250	2	7844000	250	62750	29217696	29886
    750	750	2	211219500	750	563250	319018912	164219
    1250	1250	2	977345000	1250	1563750	953030655	380058
    1750	1750	2	2681220500	1750	3064250	1995372587	597105
    2250	2250	2	5697846000	2250	5064750	3489342369	811435
    2750	2750	2	10402221500	2750	7565250	5588403544	1054141
    3250	3250	2	17169347000	3250	10565750	8295287157	1284471
    3750	3750	2	26374222500	3750	14066250	12131320544	1611701
    4250	4250	2	38391848000	4250	18066750	16618607459	1727016
    4750	4750	2	53597223500	4750	22567250	23840314552	2028285
}\dataTcLCycle

\pgfplotstableread{
    nodes	edges	non-terminals	output	max	paths	minimize-time	produce-time
    250	250	3	1500	3	750	4773822	1811
    750	750	3	4500	3	2250	37841306	3622
    1250	1250	3	7500	3	3750	97539432	3321
    1750	1750	3	10500	3	5250	139118229	3321
    2250	2250	3	13500	3	6750	229878194	3925
    2750	2750	3	16500	3	8250	374672241	3622
    3250	3250	3	19500	3	9750	514578358	3321
    3750	3750	3	22500	3	11250	649393688	3622
    4250	4250	3	25500	3	12750	842178618	3924
    4750	4750	3	28500	3	14250	1055882752	3320
}\dataPLTCycle

\pgfplotstableread{
    nodes	edges	non-terminals	output	max	paths	minimize-time	produce-time
    250	251	4	496172501	31501	31751	13626250	3835902
    750	751	4	39762423751	282001	282751	131126742	40572352
    1250	1251	4	306154300001	782501	783751	254709372	98810017
    1750	1751	4	1175046801251	1533001	1534751	505095920	196078669
    2250	2251	4	3209314927501	2533501	2535751	856655548	326975080
    2750	2751	4	7159333678751	3784001	3786751	1336229001	493089219
    3250	3251	4	13962978055001	5284501	5287751	1884782576	694917664
    3750	3751	4	24745623056251	7035001	7038751	2690009509	945633255
    4250	4251	4	40820143682501	9035501	9039751	3570266705	1202059680
    4750	4751	4	63686914933751	11286001	11290751	4320691290	1515352059
}\dataDCMatchedCycle

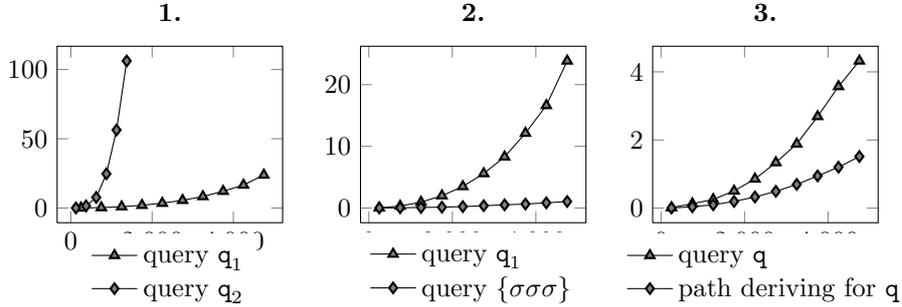
\begin{figure}[tb!]
    \centering
    \begin{tabular}{ccc}
        \begin{tikzpicture}
            \begin{axis}[plot,title={\textbf{1.}}]
                \addplot table[x=nodes,y expr=(\thisrow{minimize-time}/1000000000)] {\dataTcLCycle};
                \addplot table[x=nodes,y expr=(\thisrow{minimize-time}/1000000000)] {\dataTcNLCycle};
                \legend{\vphantom{hj}query $\NT{q}_1$,\vphantom{hj}query $\NT{q}_2$};
            \end{axis}
        \end{tikzpicture}
    &
        \begin{tikzpicture}
            \begin{axis}[plot,title={\textbf{2.}}]
                \addplot table[x=nodes,y expr=(\thisrow{minimize-time}/1000000000)] {\dataTcLCycle};
                \addplot table[x=nodes,y expr=(\thisrow{minimize-time}/1000000000)] {\dataPLTCycle};
                \legend{\vphantom{hj}query $\NT{q}_1$,\vphantom{hj}query $\{ \symbol\symbol\symbol \}$};
            \end{axis}
        \end{tikzpicture}
    &
        \begin{tikzpicture}
            \begin{axis}[plot,title={\textbf{3.}}]
                \addplot table[x=nodes,y expr=(\thisrow{minimize-time}/1000000000)] {\dataDCMatchedCycle};
                \addplot table[x=nodes,y expr=(\thisrow{produce-time}/1000000000)] {\dataDCMatchedCycle};
                \legend{\vphantom{hj}query $\NT{q}$,path deriving for $\NT{q}$};
            \end{axis}
        \end{tikzpicture}
    \end{tabular}
    \caption{Results of test measurements on Algorithm~\ref{alg:detnonrecag} and on deriving paths of minimum length. The horizontal axis displays the size of the graph ($\abs{\Nodes}$), the vertical axis displays the running time ($\mathrm{s}$), and the plot titles matches with the test descriptions in Section~\ref{sec:exp}.}\label{fig:plots}
\end{figure}

The measurements for these three tests are summarized in Figure~\ref{fig:plots}. On the one hand we see that Algorithm~\ref{alg:detnonrecag} can evaluate queries on large graphs, even if the resulting paths are large or if many paths are produced. For example, query $\NT{q}$ evaluated on double-cyclic graphs of $4750$ nodes gives a total of $11 \cdot 10^6$ paths, where the longest path consists of $11\cdot 10^6$ edges, this while the running time of Algorithm~\ref{alg:detnonrecag} is only $4.3\mathrm{s}$ and the longest path is derived from the resulting minimizing set in only $1.5\mathrm{s}$. Hence, in this case, the cost for answering query $\NT{q}$ using the single-path query semantics is at most $5.8\mathrm{s}$.

On the other hand, we see that the performance of Algorithm~\ref{alg:detnonrecag} is heavily influenced by the ambiguity of the context-free grammar. We see that query $\NT{q}_1$ evaluates magnitudes faster than query $\NT{q}_2$, even though the context-free languages underlying these queries are equivalent. The measurements on $\NT{q}_1$ and $\Lang$ show that evaluating $\Lang$ is faster, which is unsurprising as $\Lang$ is a much simpler query. Still, the difference in running time for these two queries is relatively small.

\section{Conclusions and future work}\label{sec:conclusions}

To address the limits of the traditional query semantics for navigational query languages such as the context-free path queries, we proposed path-based query semantics. We studied two such path-based query semantics, namely the all-paths query semantics and the single-paths query semantics, and we provided a formal framework for evaluating queries on graphs using both path-based query semantics. Our initial results show that the path-based query semantics have added practical value and a small-scale experiment on an implementation of the main query evaluation algorithms show that query answering is feasible, even when query results grow very large.

In conclusion, we believe that our work opens the door for further study of path-based query semantics. Besides the open problem already stated in this work---determining strict worst-case upper bounds on the size of the query result under the single-path query semantics---several other directions for future work have our interest.
\begin{enumerate}
\item Can we use more efficient algorithms for the problems outlined in this paper? Can we, for example, apply techniques used for context-free parsing~\cite{parsing_book,valiant} or for Datalog query evaluation~\cite{datalog,nowdatalog}?
\item All algorithms outlined in this paper are bottom-up. Can we derive top-down algorithms or, in general, goal-oriented algorithms for answering queries for a given pair of nodes?
\item Our measurements showed that two different context-free grammars for the same context-free language can have huge differences in the running time for query evaluation. Can we optimize context-free grammars to guarantee better performance? Can we provide more efficient query evaluation for deterministic or for unambiguous context-free grammars?
\item Are there approximation algorithms for evaluating queries using the single-path query semantics that guarantee to produce paths whose length is close to the length of paths of minimum length, while having a much lower complexity? Our initial work on this topic shows that straightforward naive methods exist to efficiently produce a deterministic non-recursive set of production rules. Although such deterministic non-recursive sets of production rules guarantee a worst-case upper bound on path lengths, the length of the resulting paths is not necessary close to optimal.
\item To which extent can we adopt path-based query evaluation such that it exploits parallel hardware, distributed computing, and/or specialized acceleration hardware?
\item Can we generalize path-based query semantics to query languages that do not query based on path structures, but query based on patterns in graphs (such as Datalog and the navigational expressions~\cite{relexpr}), and can we provide efficient query evaluation for such graph-based query semantics?
\end{enumerate}

\bibliographystyle{plain}
\bibliography{biblio}

\end{document}